%% file: Fouquere.tex
\documentclass{llncs} 

\input{STYLES/styles}

\begin{document}

\pagestyle{plain} 
\mainmatter       

\title{A Sequent Calculus for Modelling Interferences%
}

\author{Christophe Fouquer\'e}

\institute{LIPN -- UMR7030 \\ CNRS -- Universit\'e Paris 13
\\99 av. J-B Cl\'ement, F--93430 Villetaneuse, France\\
\email{cf@lipn.univ-paris13.fr}\\
}

\maketitle

\begin{abstract}
A logic calculus is presented that is a conservative extension of linear logic. 
The motivation beneath this work concerns lazy evaluation,
true concurrency and interferences in proof search.
The calculus includes two new connectives to deal with multisequent structures and has the cut-elimination property. Extensions are proposed that give first results concerning our objectives.
\end{abstract}

\section{Introduction}\label{sec:introduction}

\input{SECTIONS/Introduction}

\section{Related Works}\label{sec:RelatedWorks}

\input{SECTIONS/RelatedWorks}

\section{A Multisequent Calculus}\label{sec:calculus}

\input{SECTIONS/CalcSeq}

\section{Shared and Unshared Modalities}\label{sec:modalities}

\input{SECTIONS/csModal}

%
%
%

\section{Conclusion}\label{sec:conclusion}

\input{SECTIONS/Conclusion}

\bibliographystyle{splncs}
\bibliography{Fouquere}


\newpage
\section{Annex: Sequent Calculus}
%
\input{ANNEXE/CalcSeq_annex}
%
%

\end{document}

%% file: STYLES/styles.tex
\usepackage{latexsym}
\usepackage{amsmath}
\usepackage{amssymb}
\usepackage{STYLES/proof}
\usepackage{float} 
\floatstyle{boxed}
\newfloat{Fig.}{t}{lof}
\newfloat{Ex.}{t}{lof}

\input{STYLES/MacrosCS.tex}
\usepackage{STYLES/nfpar}

%% file: STYLES/MacrosCS.tex
\def\scriptInfer[#1]#2#3{%
\infer[\scriptstyle #1]{\scriptstyle #2}{\scriptstyle #3}%
}%
\def\scriptInferD#1#2{%
\infer{\scriptstyle #1}{\scriptstyle #2}%
}%

\newlength{\hauteurAcentrer}
\input{STYLES/lmacros.tex}

\newcommand{\plus}{\oplus}

\newcommand{\ctimes}{
\odot}

\newcommand{\cpar}{
\mid}

\newcommand{\un}{
\mathbf 1}

\newcommand{\zero}{
\mathbf 0}

\newlength{\lgctimes}
\newcommand\ConcImp{%
\mathrel{%
\relbar%
\mathchoice{\hspace{-.85\lgctimes}}{\hspace{-.85\lgctimes}}{\hspace{-.4\lgctimes}}{\hspace{-.35\lgctimes}}%
\ctimes%
}}

\newcommand{\CMALL}{CMALL}		

\newlength{\lgcirc}

\newcommand\PostImp{\settowidth{\lgcirc}{$\circ$}
\mathrel{
\relbar
\mathchoice{\hspace{-.85\lgcirc}}{\hspace{-.85\lgcirc}}{\hspace{-.4\lgcirc}}{\hspace{-.35\lgcirc}}

\circ
}}

\newlength{\Mentrylength}
\newlength{\Mentryheight}
\newcommand{\Mentrylabel}[1]%
	{\settowidth{\Mentrylength}{\textsf{#1\raisebox{\Mentryheight-1.5ex}{ : }}}%
	\settoheight{\Mentryheight}{\textsf{#1\raisebox{\Mentryheight-1.5ex}{ : }}}
	\ifthenelse{\lengthtest{\Mentrylength > 60pt}}%
		{\setlength{\labelwidth}{\Mentrylength+5pt}}%
		{\setlength{\labelwidth}{65pt}}%
	\raisebox{1.5ex-\Mentryheight}[1ex][\Mentryheight]{\makebox[\labelwidth]{%
		{\parbox[t]{\labelwidth}{\textsf{#1\raisebox{\Mentryheight-1.5ex}{ : }}}}}}%
}
{\begin{list}{}%
	{%
	\setlength{\labelwidth}{40pt}%
	\setlength{\leftmargin}{\labelwidth+\labelsep}%
	}%
}%
{\end{list}}

%% file: STYLES/lmacros.tex
\newcommand{\la}{\mbox{$-\!\circ$}} 

\newcommand{\Par}{\mathrel{\wp}}

\newcommand{\with}{\,{\&}\,}

\newcommand{\stred}[1]{\mbox{$\Longrightarrow$}}

\newcommand{\stredstar}[1]{\mbox{$\Longrightarrow^*$}}

%% file: SECTIONS/Introduction.tex

Linear Logic is a good framework for interpreting and computing over linear structures.
Since Girard's seminal paper~\cite{Girard87} that gives first insights (proof nets, phase and coherent spaces), a lot has been achieved among which normalization of proofs via focusing and polarization~\cite{Andreoli92,Laurent05a}. These last results seem to be intrinsically related to principles underlying cut elimination as it allows for investigating a reconstruction of logical structures as in Ludics~\cite{Girard01}.
Recent works done on concurrent modelling using such a framework seem promising~\cite{DBLP:conf/csl/CurienF05,DBLP:conf/csl/GiamberardinoF06}.
However, non series-parallel situations are not taken into account.

We present a logic calculus (and variants) that is a conservative extension of linear logic. 
The motivation beneath this work is a careful study of lazy evaluation in logic programming. Since works of Andreoli~\cite{Andreoli92}, we know that full linear logic may be used as a logical programming language thanks to focalization and works have been done on lazy evaluation in this case~\cite{DBLP:conf/elp/CervesatoHP96}. However we show in Sect.~\ref{sec:calculus} that cut elimination is false for a naive calculus taking laziness as a principle. A second motivation concerns the control of true concurrency and interferences in proof search. For instance, suppose the following problem to be modelled in logic programming. We have two 'packs' of actions: $f = \bigoplus f_n$ (resp. $g = \bigoplus g_n$) such that $f_n$ (resp. $g_n$) transforms $n$ occurences of $a$ (resp. $b$) into $n$ occurences of $a$ and $n$ occurences of $b$, where $n\geq 1$. We suppose at the initial state only one resource
of each kind (hence one $a$ and one $b$).
We want to simulate exactly the two following situations: 
\begin{itemize}
\item[(i)] if the two actions are applied (whatever may be the order) then we have three possible results:
3 $a$ and 2 $b$, 2 $a$ and 3 $b$, 2 $a$ and 2 $b$. The first (resp. second) result is obtained when action $f$ (resp. $g$) is applied first followed by action $g$ (resp. $f$). The third result occurs when the two actions are performed independently.
\item[(ii)] if the two actions are applied strictly concurrently, there is only one possible result: we get 2 $a$ and 2 $b$.
\end{itemize}
This is not possible inside propositional classical or linear logic as it requires a control between proofs.
For that purpose, we basically shift from a sequent view to a multisequent view. Moreover sharing of formulas occurences between such sequents is allowed. The reader should have in mind the following elements:
\begin{itemize}
\item logical operations are done on occurences of formulas that may be shared among different sequents,
\item a sequent is a place grouping a bunch of occurences,
\item each step of a proof transforms zero, one or two multisequents into one by means of an operation, either structural or logical (in this last case the operation is done on occurences of formula and entails the structure of the conclusion),
\item equivalently, a multisequent may be defined as a set of places, a set of occurences of formulas and a function relating a place to a set of occurences.
\end{itemize}
In the following the two interpretations may be used. 
Shifting from sequents to multisequents gives place for a new structural operation that joins sequents (to be compared with the par operation that joins occurences in a sequent). In a first step we consider a "2-way" connective that "relates" two sequents in a multisequent. Its dual is denoted $\cpar$ and called cpar.
We then consider the following extensions:
\begin{itemize}
\item add of a 'cloning' structural rule: this comes from the observation that interaction of a sequent by means of a cut elimination is behaviouraly equivalent to interaction with two sequents sharing exactly the same occurences. However this last observation is not provable without such a cloning rule.
\item add exponential-like modalities (Sec.~\ref{sec:modalities}): standard modalities for linear logic are available. However as sharing is internal, it allows for adding modalities whose behaviour is the converse of the standard one.
\end{itemize}

%% file: SECTIONS/RelatedWorks.tex

Modelling interferences has not yet been really investigated in logic. First of all, classical logic as well as modal logic do not take seriously into account the notion of resource, hence appear to be inadequate. Second, modelling (and controlling) interferences may seem contradictory in the framework of Linear Logic as the splitting mechanism seems at the heart of cut-elimination. However, current works done on concurrency are close.
Following Girard's works on Ludics, Curien, Faggian, Giamberardino~\cite{DBLP:conf/csl/CurienF05,DBLP:conf/csl/GiamberardinoF06} were able to formalize L-nets that quotient (abstract) proof trees w.r.t. commutation of tensors. This normalization goes further than the one given by Andreoli with focusing and polarization.
However, non series-parallel situations cannot be taken into account in their denotational model.
Works close to the research presented here include Bunched Implications~\cite{DBLP:journals/bsl/OHearnP99} and Deep Inference~\cite{DBLP:journals/tocl/Guglielmi07}.
But these two last frameworks seem to fail in keeping basic logical properties as focalization and polarization.
The line of research that is undertaken here introduces a syntactic novelty by considering that occurences of formulae may be shared by different sequents. This sharing induces a strict synchronization between different computations (i.e. developments of proofs) and new connectives may be defined that internalize this mechanism.

%% file: SECTIONS/CalcSeq.tex

Besides the classical multiplicative and additive connectives of Linear Logic, we introduce two new connectives ctimes $\ctimes$ and cpar $\cpar$ whose intended meaning is to model strict concurrency.

\begin{definition}
The {\em formulas}, denoted $A, B, \dots$, are built from atoms $p$, $q$, 
$\dots$, $p^{\perp}$, $q^{\perp}$, $\dots$,
constants $\un$, $\bot$, $\zero$, $\top$ and the following (linear) connectives:
\begin{itemize}
\item (parallel) multiplicative conjunction $\otimes$ (times) and 
disjunction $\Par$ (par),
\item (concurrent) multiplicative conjunction $\ctimes$ (ctimes) and 
disjunction $\cpar$ (cpar),
\item additive conjunction $\plus$ (plus) and disjunction $\with$ (with).
\end{itemize}

Negation is defined by De Morgan rules:
\center$\begin{array}{lp{.5cm}l}
(p)^{\perp} = p^{\perp}		&&	(p^{\perp})^{\perp} = p	\\
(A \otimes B)^{\perp} = B^{\perp} \Par A^{\perp}	&&	(A \Par B)^{\perp} = 
B^{\perp} \otimes A^{\perp} \\
(A \ctimes B)^{\perp} = B^{\perp} \cpar A^{\perp}	&&	(A \cpar B)^{\perp} = 
B^{\perp} \ctimes A^{\perp} \\
(A \plus B)^{\perp} = B^{\perp} \with A^{\perp}		&&	(A \with B)^{\perp} = 
B^{\perp} \plus A^{\perp} \\
\un^{\perp} = \bot ~~~~	\zero^{\perp} = \top 
&&	\bot^{\perp} = \un ~~~~ \top^{\perp} = \zero
\end{array}$

$\begin{array}{c}
A \PostImp B = A^{\perp} \Par B
\end{array}$
\end{definition}

\subsection{Structures of Multisequents}
\begin{definition}
A {\em formula context} $\Gamma$ has one of the two following forms:
\begin{itemize}
\item A formula $A$
\item A finite multiset of formula contexts separated by commas $\Delta_1, \dots, 
\Delta_n$.
',' is considered commutative and associative.
\end{itemize}
{\em Sequents} are of the form $\{\Gamma\}$, where $\Gamma$ is a formula context.
A {\em multisequent} is a finite multiset of sequents. 
Multisequents are denoted $\mathcal S, \mathcal T, \dots$ If a multisequent is reduced to one sequent, '$\{$' and '$\}$' may be omitted.
A multisequent may contain sequents that are not disjoint: an occurence of a formula may appear in different sequents. Such sequents
are said to be {\em linked}. Superscripts are put if different occurences of the same formula occur.
The {\em principal} formulas occurences of a (logical) rule are formulas from the hypotheses on which the rule applies. The {\em principal} sequents are sequents where the principal formulas occur.
\end{definition}

\begin{example}[multisequents]
$\{A,B\}\{C,D\}$: this multisequent involves four formulas and two (disjoint) sequents whereas
$\{A,B\}\{B,C\}\{C,D\}$ involves three sequents and four occurences of formulas. \end{example}

At first glance, rules given in sequent calculi may seem strange:
Throughout the paper, contexts of a principal occurence are identified by a free subscript. For example, 
$\{\Gamma_i, A\}$ means a multiset of sequents (the domain of the free subscript $i$) where the
{\em same} occurence $A$ appears. Note that the domain of $i$ cannot be empty.
Sequents that remain unchanged by a rule are either replaced by dots or by a notation for a multisequent. 
If a proof involves different occurences
of the same fomula (hence different contexts), these occurences are distinguished by a superscript:\footnote{A context $\Gamma$ with a supersript supposes the superscript for
each formula of the context.} remark this in the $\with$-rule in Fig.~\ref{fig:seq-imp}.

\subsection{A Naive (and Wrong) Attempt}

Lazy logic programming relies mainly on a lazy splitting of contexts when considering the $\otimes$ rule. The standard $\otimes$ rule is the following one:
$$
\infer[(\otimes)]{\{\Delta, \Gamma, A\otimes B\}}
                 { \{\Gamma, A\}
	             &
	            \{\Delta, B\}}
$$
In a bottom-up proof search, as it is the case in logic programming, applying this rule requires to know how to split the multiset $\Delta, \Gamma$. A lazy way consists in delaying this separation. Let us note $\cpar$ the connective $\otimes$ defined in a lazy way. Shifting to multisequents, this may be given by sharing the whole multiset $\Delta, \Gamma$ between the two sequents in the hypothesis (remember that occurences are shared between sequents if no superscript is present):
$$
\infer[(\cpar)]{\{\Delta, \Gamma, A\cpar B\}}
                 { \{\Delta, \Gamma, A\}
	           \{\Delta, \Gamma, B\}}
$$

We suppose further $\Par$ still dual to $\cpar$: $(A \cpar B)^{\perp} = B^{\perp} \Par A^{\perp}$. Following these guidelines, a system for a lazy Multiplicative Linear Logic (lazy MLL) is given in Fig.~\ref{badCLL}. However a counter-example to cut-elimination is easy to find:\\
$\{A^\perp, A\Par \bot\}$ ~~~and~~~ 
$\{A^\perp \cpar 1, A \Par \bot\}\{A^\perp \cpar 1, 1^1 \Par 1^2\}$ ~~~are provable.\\ 
But~~~~ 
$\{A^\perp, A \Par \bot\}\{A^\perp, 1^1 \Par 1^2\}$ ~~~~is not provable:

\settoheight{\hauteurAcentrer}{\scriptInfer[]{\{A^\perp \cpar 1, A \Par \bot\}\{A^\perp \cpar 1, 1^1 \Par 1^2\}}{
	\scriptInfer[]{\{A^\perp \cpar 1, A, \bot\}\{A^\perp \cpar 1, 1^1, 1^2\}}{
		\scriptInfer[]{\{A^\perp \cpar 1, A\}\{A^\perp \cpar 1, 1^1, 1^2\}}{
			\scriptInfer[]{\{A^\perp, A\}\{A^\perp, 1^1, 1^2\}\{1, A\}\{1, 1^1, 1^2\}}{
				\scriptInfer[]{\{A^\perp, A\}\{1^1\}\{1\}\{1^2\}}{
				}
			}
		}
	}
}
}%

$$
\raisebox{.25\hauteurAcentrer}{\scriptInfer[]{A^\perp, A\Par \bot}{
	\scriptInfer[]{A^\perp, A, \bot}{
		\scriptInfer[]{A^\perp, A}{
		}
	}
}}
\raisebox{.5\hauteurAcentrer}{\mbox{ ~~~and~~~ }} 
\scriptInfer[]{\{A^\perp \cpar 1, A \Par \bot\}\{A^\perp \cpar 1, 1^1 \Par 1^2\}}{
	\scriptInfer[]{\{A^\perp \cpar 1, A, \bot\}\{A^\perp \cpar 1, 1^1, 1^2\}}{
		\scriptInfer[]{\{A^\perp \cpar 1, A\}\{A^\perp \cpar 1, 1^1, 1^2\}}{
			\scriptInfer[]{\{A^\perp, A\}\{A^\perp, 1^1, 1^2\}\{1, A\}\{1, 1^1, 1^2\}}{
				\scriptInfer[]{\{A^\perp, A\}\{1^1\}\{1\}\{1^2\}}{
				}
			}
		}
	}
}
\raisebox{.5\hauteurAcentrer}{\mbox{ ~~~but~~~ }} 
\scriptInfer[]{\{A^\perp, A \Par \bot\}\{A^\perp, 1^1 \Par 1^2\}}{
	\scriptInfer[]{\{A^\perp, A, \bot\}\{A^\perp, 1^1, 1^2\}}{
		\scriptInfer[]{\{A^\perp, A\}\{A^\perp, 1^1, 1^2\}}{
			\scriptInfer[]{\{A^\perp, A\}\{1^1, 1^2\}}{
				\scriptInfer[]{\{1^1, 1^2\}}{ \textit{false}
				}
			}
		}
	}
}
$$

\begin{Fig.}
{\em Structural rules}
\[
\begin{array}{lclcl}
\infer[(d)]{\dots \{\Gamma_i, A\} \{A, \Delta\} \dots}
           {\dots \{\Gamma_i, A\} \{\Delta\} \dots}
&~~~~~~&
\infer[(s)]{{\mathcal S}_1 {\mathcal S}_2}
      {{\mathcal S}_1
      &
      {\mathcal S}_2}
&~~~~~~&
\infer[(e)]{\dots \{\Gamma, B, A, \Delta \} \dots}
      {\dots \{\Gamma, A, B, \Delta \} \dots}
\end{array}
\]

{\em Logical rules} (in rules $(1)$ and $(axiom)$, the multisequent consists of only one sequent)
\[
\begin{array}{ccc}
\multicolumn{3}{c}{\infer[(axiom)]{A, A^\perp}{}}
\\ \\
\infer[(1)]{1}{} 
&~~~~~~&
\infer[(\bot)]{\dots \{\Gamma_i,\bot\} \dots}
      {\dots \{\Gamma_i\} \dots}
\\ \\
\infer[(\Par)]{\dots \{\Gamma_i,A\Par B\} \dots}
              {\dots \{\Gamma_i,A, B\} \dots} 
&~~~~~~&
\infer[(\cpar)]{\dots \{\Gamma_i, A\cpar B\} \dots}
     {\dots \{\Gamma_i, A\} \{\Gamma_i,B\} \dots}
\end{array}
\]

{\em Cut rule}
\[
\infer[(cut)]{\dots \{\Gamma_i, \Delta_j\} \dots}
    {
    \dots \{\Gamma_i, A\} \dots
    &
    \dots \{\Delta_j, A^\perp\} \dots}
\]
\caption{Sequent calculus for a bad lazy MLL. $i,j \in \mathbb{N}^*$ in rules.}\label{badCLL}
\end{Fig.}

\subsection{The Calculus \CMALL}

In order to circumvent the previous situation, lazyness is modelled by means of two specific connectives $\cpar$ and $\ctimes$ besides the two multiplicative connectives $\Par$ and $\otimes$ of Linear Logic. The rules of the sequent calculus Concurrent Multiplicative Additive Linear Logic (\CMALL) are given in Fig.~\ref{fig:seq-imp}. The system includes a cloning structural rule (c), however one may note that proofs of cut elimination, asynchrony, ... we give in the following are still true without this rule. Examples of instantiation of the rules are given below to help the reader recover standard situations.

\begin{example}[Rule instantiation]
($A, B, X, Y, Z$ are formulas)\\
\[
\begin{array}{ccc}
\infer[(\otimes)]{\{X,Z,A\otimes B\}\{Y,Z,A\otimes B\}}{
	\{X,A\}\{Y,A\}
	&
	\{Z,B\}
}
&
\infer[(\cpar)]{\{X\}\{X,A\cpar B\}}{
	\{X\}\{X,A\}\{X,B\}
}
&
\infer[(\Par)]{\{X\}\{X,A\Par B\}}{
	\{X\}\{X,A,B\}
}
\end{array}
\]
\end{example}

It is easy to prove the following statements (multisequents may be given two-sided for easiness of reading):
\begin{itemize}
\item $A \otimes B \la A \cpar B$ is provable:\vspace*{-3mm}\\
\centerline{$
\scriptInfer[\Par]{\{A^\perp \Par B^\perp, A \cpar B\}}
      {\scriptInfer[\cpar]{\{A^\perp, B^\perp, A \cpar B\}}
             {\scriptInfer[d]{\{A^\perp, B^\perp, A\}\{A^\perp, B^\perp, B\}}
                    {\scriptInfer[d]{\{A^\perp, B^\perp, A\}\{B^\perp, B\}}
                           {\scriptInfer[s]{\{A^\perp, A\}\{B^\perp, B\}}
                                  {\scriptInfer[ax]{\{A^\perp, A\}}
                                             {}
                                  &
                                  \scriptInfer[ax]{\{B^\perp, B\}}
                                             {}
      }       }     }      }      }
$}

\item  $\cpar$ is asynchronous (lemma~\ref{lemma:asynchrony}) whereas $\otimes$ is synchronous. Although $\cpar$ does neither distribute over $\Par$, nor the converse. But $\cpar$ does distribute over $\with$: $A \cpar (B \with C) \dashv\vdash (A \cpar B) \with (A \cpar C)$ is provable\vspace*{1mm}\\
\hspace*{-2cm}$
	\scriptInferD{\{A \cpar (B \with C) , (A^\perp \ctimes B^\perp) \plus (A^\perp \ctimes C^\perp)\}}{
		\scriptInferD{\{A,(A^\perp \ctimes B^\perp) \plus (A^\perp \ctimes C^\perp)\}~~\{B \with C, (A^\perp \ctimes B^\perp) \plus (A^\perp \ctimes C^\perp)\}}{
			\scriptInferD{\{A^1,(A^\perp \ctimes B^\perp) \plus (A^\perp \ctimes C^\perp)^1\}~~\{B, (A^\perp \ctimes B^\perp) \plus (A^\perp \ctimes C^\perp)^1\}}{
				\scriptInferD{\{A^1,(A^\perp \ctimes B^\perp)^1\}~~\{B, (A^\perp \ctimes B^\perp)^1\}}{
					\scriptInferD{\{A^1,A^{\perp 1}\}}{
					}
					&
					\scriptInferD{\{B, B^{\perp 1}\}}{
					}
				}
			}
			&
			\scriptInferD{\{A^2,(A^\perp \ctimes B^\perp) \plus (A^\perp \ctimes C^\perp)^2\}~~\{C, (A^\perp \ctimes B^\perp) \plus (A^\perp \ctimes C^\perp)^2\}}{
				\scriptInferD{\{A^2,(A^\perp \ctimes C^\perp)^2\}~~\{C, (A^\perp \ctimes C^\perp)^2\}}{
					\scriptInferD{\{A^2,A^{\perp 2}\}}{
					}
					&
					\scriptInferD{\{C,C^{\perp 2}\}}{
					}
				}
			}	
		}
	}
$
	
	$$\scriptInferD{\{(A \cpar B) \with (A \cpar C), A^\perp \ctimes (B^\perp \plus C^\perp)\}}{
		\scriptInferD{\{A \cpar B, [A^\perp \ctimes (B^\perp \plus C^\perp)]^1\}}{
			\scriptInferD{\{A, [A^\perp \ctimes (B^\perp \plus C^\perp)]^1\}~~\{B, [A^\perp \ctimes (B^\perp \plus C^\perp)]^1\}}{
				\scriptInferD{\{A, A^{\perp 1}\}}{
				}
				&
				\scriptInferD{\{B, (B^\perp \plus C^\perp)^1\}}{
					\scriptInferD{\{B, B^{\perp 1}\}}{
					}
				}
			}
		}
		&
		\scriptInferD{\{A \cpar C, [A^\perp \ctimes (B^\perp \plus C^\perp)]^2\}}{
			\scriptInferD{\{A, [A^\perp \ctimes (B^\perp \plus C^\perp)]^2\}~~\{C, [A^\perp \ctimes (B^\perp \plus C^\perp)]^2\}}{
				\scriptInferD{\{A, A^{\perp 2}\}}{
				}
				&
				\scriptInferD{\{C, (B^\perp \plus C^\perp)^2\}}{
					\scriptInferD{\{C, C^{\perp 2}\}}{
					}
				}
			}
		}
	}$$

\item $\otimes \not\equiv \cpar$
Remark that $\{\un, \bot \otimes \bot\}$ is not provable, but $\{\un, 
\bot \cpar \bot\}$ is provable (the two $\bot$ are indexed to distinguish 
them, however these two denote the same constant; note also that there is 
only one occurence of $\un$ throughout the proof):
\vspace{-.4cm}
$$
\scriptInfer[\cpar]{\un,\bot^1 \cpar \bot^2}
     {\scriptInfer[\bot]{\{\un,\bot^1\}\{\un, \bot^2\}}
            {\scriptInfer[\bot]{\{\un\}\{\un, \bot^2\}}
                   {\scriptInfer[w]{\{\un\}\{\un\}}
                          {\scriptInfer[\un]{\un}
                                    {}
                          }
                   }
            }
     }
$$
\end{itemize}


\begin{proposition}
The system enjoys cut-elimination: if $\mathcal S$ is a provable multisequent, 
then there exists at least one cut-free proof of $\mathcal S$.
\end{proposition}

The proof of cut-elimination (see annex) relies mainly on a reconstruction of proofs in case the two last rules concern the cut formulas, and on the three following lemmas that allow the commutation of rules. The standard definition of the height of a proof is generalized: the height of the proof of a multisequent is the
maximum of the heights of each partial proof.

\begin{lemma}[Separability]
Let $\mathcal S$ and $\mathcal T$ be disjoint multisequents (i.e. there is no occurence of formulas shared by $\mathcal S$ and $\mathcal T$), the multisequent $\mathcal S \mathcal T$ is provable iff $\mathcal S$ is provable and $\mathcal T$ is provable. 
\end{lemma}

\begin{lemma}[Asynchrony]\label{lemma:asynchrony}
The connectives $\Par, \with, \cpar$ are asynchronous: let R be an 
inference rule of one of these connectives (denoted $\circ$ below), let 
$\mathcal S$ be a provable sequent of proof
$$
\infer{\dots \{A\circ B, \Gamma\} \dots}
    {\infer{\dots \: \vdots \: \dots}
         {\infer[\mbox{R on }A\circ B]{\dots \{A\circ B, \Gamma\} \dots}
              {\mathcal T}
          }
     }
$$
then there exists a proof of the same height of $\mathcal S$ with R as the last rule.
\end{lemma}

\begin{lemma}[Synchrony of the cut rule]
The cut rule is synchronous, i.e. let a proof of $\mathcal S$ be of the form in the left hand side (R is a rule), then one can build a proof of the same height of $\mathcal S$ of the form in the right hand side:
\vspace{-.2cm}
$$
\begin{array}{ccc}
\infer[cut]{\mathcal S}
	{
	\infer[R]{{\mathcal W}_1[A]}
		{{\mathcal U}[A]}
	&
	{\mathcal V}[A^\perp]
	}
& \hspace{3cm} &
\infer[R]{\mathcal S}
	{
	\infer[cut]{{\mathcal W}_2}
		{
		{\mathcal U}[A]
		&
		{\mathcal V}[A^\perp]
		}
	}
\end{array}
$$
\end{lemma}

%

\begin{Fig.}
{\em Structural rules}
\[
\begin{array}{lclcl}
\infer[(c)]{\dots \{\Delta\} \{\Delta\} \dots}
           {\dots \{\Delta\} \dots}
&~~~~~~&
\infer[(d)]{\dots \{\Gamma_i, A\} \{A, \Delta\} \dots}
           {\dots \{\Gamma_i, A\} \{\Delta\} \dots}
&~~~~~~&
\infer[(s)]{{\mathcal S}_1 {\mathcal S}_2}
      {{\mathcal S}_1
      &
      {\mathcal S}_2}
\end{array}
\]

{\em Logical rules} (in rules $(1)$ and $(axiom)$, the multisequent consists of only one sequent)
\[
\begin{array}{cc}
\multicolumn{2}{c}{\infer[(axiom)]{A, A^\perp}{}}
\\ \\
\infer[(1)]{1}{} 
&
\infer[(\bot)]{\dots \{\Gamma_i,\bot\} \dots}
      {\dots \{\Gamma_i\} \dots}
\\ \\
\infer[(\top)]{\dots \{\Gamma_i, \top\} \dots}
     {\dots} 
&
\mbox{no rule for 0}
\\ \\
\infer[(\otimes)]{\dots \{\Delta_j, \Gamma_i, A\otimes B\} \dots}
                 { \dots \{\Gamma_i, A\} \dots
	             &
	              \dots \{\Delta_j, B\} \dots }
&
\infer[(\Par)]{\dots \{\Gamma_i,A\Par B\} \dots}
              {\dots \{\Gamma_i,A, B\} \dots} 
\\ \\
\infer[(\ctimes)]{\dots \{\Gamma_i, A\ctimes B\} 
                        \{\Delta_j, A\ctimes B\} 
                 \dots}
      {\dots \{\Gamma_i, A\} \dots
	  &
	   \dots \{\Delta_j, B\} \dots}
&
\infer[(\cpar)]{\dots \{\Gamma_i, A\cpar B\} \dots}
     {\dots \{\Gamma_i, A\} \{\Gamma_i,B\} \dots}
\\ \\
\begin{array}{ll}
\infer[(\plus_1)]{\dots \{\Gamma_i, A\plus B\} \dots}
     {\dots \{\Gamma_i, A\} \dots} 
&
\infer[(\plus_2)]{\dots \{\Gamma_i, A\plus B\} \dots}
     {\dots \{\Gamma_i, B\} \dots}
\end{array} &
\infer[(\with)]{{\mathcal S} \{\Gamma_i, A\with B\}}
    {
	{\mathcal S}^1 \{\Gamma_i^1, A\}
	&
	{\mathcal S}^2 \{\Gamma_i^2,B\}}
\end{array}
\]

{\em Cut rule}
\[
\infer[(cut)]{\dots \{\Gamma_i, \Delta_j\} \dots}
    {
    \dots \{\Gamma_i, A\} \dots
    &
    \dots \{\Delta_j, A^\perp\} \dots}
\]
\caption{Sequent calculus for \CMALL}
\label{fig:seq-imp}
\end{Fig.}

%% file: SECTIONS/csModal.tex

Modalities may be added to the system in the spirit of exponentials in Soft Linear Logic~\cite{DBLP:journals/tcs/Lafont04}.
They are written as upperscripts on formulas: $A^s$ and $A^u$.
The sharing $.^s$ modality (resp. the unsharing $.^u$) is reminiscent of the why-not $?$ (resp. the of-course $!$). Rules are completed with the ones given below:

$$
\begin{array}{lcl}
\infer[(.^s)]{\dots \{\Gamma_i, \Delta_j^s\} \dots}
		   {\dots \{\Gamma_i^j, \Delta_j\} \dots}
&~~~~~~~~&
\infer[(.^u)]{\dots \{\Gamma_i^s, A^u\} \dots}
	{\dots \{\Gamma_i, A\} \dots}
\end{array}
$$

\begin{proposition}
Cut-elimination for \CMALL\ with modalities is valid.
\end{proposition}

\begin{example}(Rule instantiation)
\[
\begin{array}{lcl}
\scriptInfer[(.^s)]{\{X, A^s\} \{Y, A^s\} \{X, B^s\} \{Y, B^s\}}
		   {\{X^1, A\} \{Y^1, A\} \{X^2, B\} \{Y^2, B\}}
&~~~~~~~~&
\scriptInfer[(.^u)]{\{X^s, A^u\} \{Y^s, A^u\}}
	{\{X, A\} \{Y, A\}}
\end{array}
\]
\end{example}


The sharing modality enjoys the following property: $A \la A^s \cpar A^s$ is provable
$$
\scriptInfer[\cpar]{A^\perp, A^s \cpar A^s}
	{\scriptInfer[.^s]{\{A^\perp, A^s\} \{A^\perp, A^s\}}
		{ \scriptInfer[{\rm ax}]{A^{\perp 1}, A^1}{}
		&
		  \scriptInfer[{\rm ax}]{A^{\perp 2}, A^2}{}
		}
	}
$$

The previous example shows that a unique resource $A$ may be used for two different actions: let us suppose a system has one resource $A$, and a set of processes each needing 
one resource $A$, may we run them together ? The answer is yes if two conditions
are satisfied: (i) each process accepts to share its needed resource with others,
(ii) the processes run concurrently.
We formalize each process ${\mathcal P}_i$ as $A^s \ConcImp R_i$ where $A \ConcImp B = A^\perp \cpar B$ ($R_i$ is
the formula modelling the result of ${\mathcal P}_i$): this answers condition (i).
Concurrence between processes is denoted as ${\mathcal P}_1 \ctimes \dots \ctimes {\mathcal P}_n$.
We have then the following provable and non-provable two-sided sequents ($1\le i\le n$):
$$
A, {\mathcal P}_i \vdash R_i
$$
$$
A, {\mathcal P}_1 \ctimes \dots \ctimes {\mathcal P}_n \vdash R_1 \ctimes \dots \ctimes R_n
$$
$$
A, {\mathcal P}_1 \otimes \dots \otimes {\mathcal P}_n \not\vdash R_1 \otimes \dots \otimes R_n 
$$

The fact that the third sequent is not provable is obvious (even if each process is modelled $A \la R_i$ !). We just give the proofs for the two others (we set $n = 2$ in the
second proof for sake of clarity).

$$
\scriptInfer[\otimes]{A^\perp, A^s\otimes R_i^\perp, R_i}
	{ \scriptInfer[{\rm ax}]{R_i^\perp, R_i}{}
	&
	  \scriptInfer[(.^s)]{A^\perp, A^s}
		{\scriptInfer[{\rm ax}]{A^\perp, A}{}}
	}
$$

$$
\scriptInfer[\cpar]{A^\perp, (A^s\cpar R_1^\perp) \cpar (A^s\cpar R_2^\perp), R_1 \ctimes R_2}
	{\scriptInfer[\cpar]{	\{A^\perp, A^s\cpar R_1^\perp, R_1 \ctimes R_2\}
				\{A^\perp, A^s\cpar R_2^\perp, R_1 \ctimes R_2\}}
		{\scriptInfer[\cpar]{	\{A^\perp, A^{s 1}, R_1 \ctimes R_2\}
					\{A^\perp, R_1^\perp, R_1 \ctimes R_2\}
					\{A^\perp, A^s\cpar R_2^\perp, R_1 \ctimes R_2\}}
			{\scriptInfer[d]{	\{A^\perp, A^{s 1}, R_1 \ctimes R_2\}
						\{A^\perp, R_1^\perp, R_1 \ctimes R_2\}
						\{A^\perp, A^{s 2}, R_1 \ctimes R_2\}
						\{A^\perp, R_2^\perp, R_1 \ctimes R_2\}}
				{\scriptInfer[d]{	\{A^\perp, A^{s 1}, R_1 \ctimes R_2\}
							\{R_1^\perp, R_1 \ctimes R_2\}
							\{A^\perp, A^{s 2}, R_1 \ctimes R_2\}
							\{A^\perp, R_2^\perp, R_1 \ctimes R_2\}}
					{\scriptInfer[d]{	\{A^\perp, A^{s 1}\}
								\{R_1^\perp, R_1 \ctimes R_2\}
								\{A^\perp, A^{s 2}, R_1 \ctimes R_2\}
								\{A^\perp, R_2^\perp, R_1 \ctimes R_2\}}		
						{\scriptInfer[d]{	\{A^\perp, A^{s 1}\}
									\{R_1^\perp, R_1 \ctimes R_2\}
									\{A^\perp, A^{s 2}\}
									\{A^\perp, R_2^\perp, R_1 \ctimes R_2\}}		
							{\scriptInfer[s]{	\{A^\perp, A^{s 1}\}
										\{R_1^\perp, R_1 \ctimes R_2\}
										\{A^\perp, A^{s 2}\}
										\{R_2^\perp, R_1 \ctimes R_2\}}
								{ \scriptInfer[(.^s)]{	\{A^\perp, A^{s 1}\}
											\{A^\perp, A^{s 2}\}}
									{\scriptInfer[s]{	\{A^{\perp 1}, A^{1}\}
												\{A^{\perp 2}, A^{2}\}}
										{ \scriptInfer[{\rm ax}]{\{A^{\perp 1}, A^{1}\}}{}
										&
										  \scriptInfer[{\rm ax}]{\{A^{\perp 2}, A^{2}\}}{}
										}
									}
								&
							  	  \scriptInfer[\ctimes]{	\{R_1^\perp, R_1 \ctimes R_2\}
												\{R_2^\perp, R_1 \ctimes R_2\}}
									{ \scriptInfer[{\rm ax}]{\{R_1^\perp, R_1\}}{}
									&
									  \scriptInfer[{\rm ax}]{\{R_1^\perp, R_1\}}{}
	}	}	}	}	}	}	}	}	}
$$

%% file: SECTIONS/Conclusion.tex

An original logic calculus (with variants) is presented that is a conservative extension of Linear Logic, at the theoretical level, and at the language level. 
The motivation beneath this work concerns lazy evaluation,
true concurrency and interferences in proof search.
We show that cut elimination is false if one considers a naive approach. The calculus \CMALL\ adds two new connectives to deal with multisequent structures. It has the cut-elimination property. Extensions are proposed that give first results concerning our objectives.

%% file: ANNEXE/CalcSeq_annex.tex


We only give sketches of the proofs.

\begin{lemma}[Separability]
Let $\mathcal S$ and $\mathcal T$ be disjoint multisequents (i.e. there are no occurences of formulae appearing in $\mathcal S$ and 
in $\mathcal T$), the multisequent $\mathcal S \mathcal T$ is provable iff $\mathcal S$ is provable and $\mathcal T$ is provable. 
\end{lemma}

\begin{proof}
The structural rule of separation $(s)$ gives one direction. The other direction results from the following remark: principal sequents give linked sequents in the conclusion, except for the rule of separation $(s)$. Hence rules apply independently on $\mathcal S$ and $\mathcal T$.
\end{proof}



\begin{lemma}[Asynchrony]
The connectives $\Par, \with, \cpar$ are asynchronous: let R be an 
inference rule of one of these connectives (noted $\circ$ below), let 
$\mathcal S$ be a provable sequent of proof
$$
\scriptInfer[]{\dots \{A\circ B, \Gamma\} \dots}
    {\scriptInfer[]{\dots \: \vdots \: \dots}
         {\scriptInfer[\mbox{R on }A\circ B]{\dots \{A\circ B, \Gamma\} \dots}
              {\mathcal T}
          }
     }
$$
then there exists a proof of $\mathcal S$ with R as the last rule.
\end{lemma}

\begin{proof}
It suffices to prove that a rule may be shifted upward if the before last rule
concerns an asynchronous connective. This is proved by induction on the 
height of the proof.

\begin{description}
\item[\fbox{case $\Par$}] \ \\
     \begin{description}
     \item[rule $(w)$: ] From ($A$ and $B$ range over $i$ and $i_0$ in the following proof):
          $$\scriptInfer[w]{\dots \{A \Par B, \Gamma_i\} 
                            \{A \Par B, \Gamma_{i_0}\} 
                            \{A \Par B, \Gamma_{i_0}\} \dots}
                  {\scriptInfer[\Par]{\dots \{A \Par B, \Gamma_i\} 
                                      \{A \Par B, \Gamma_{i_0}\} \dots}
                         {\dots \{A, B, \Gamma_i\} 
                                \{A, B, \Gamma_{i_0}\} \dots}
                   }
          $$
          then one can define the following proof:
          $$\scriptInfer[\Par]{\dots \{A \Par B, \Gamma_i\} 
                               \{A \Par B, \Gamma_{i_0}\} 
                               \{A \Par B, \Gamma_{i_0}\} \dots}
                  {\scriptInfer[w]{\dots \{A, B, \Gamma_i\} 
                                   \{A, B, \Gamma_{i_0}\} 
                                   \{A, B, \Gamma_{i_0}\} \dots}
                         {\dots \{A, B, \Gamma_i\} 
                                \{A, B, \Gamma_{i_0}\} \dots}
                   }
          $$
     \item[rule $(d)$: ] From:
          $$\scriptInfer[d]{\dots \{\Gamma_i, A\Par B\} 
                            \{\Delta_j, A\Par B\} \dots}
               {\scriptInfer[\Par]{\dots \{\Gamma_i, A\Par B\} 
                                   \{\Delta_j\} \dots}
                    {\dots \{\Gamma_i, A, B\} 
                           \{\Delta_j\} \dots}
               }
          $$
          One can build (twice rule d):
          $$\scriptInfer[\Par]{\dots \{\Gamma_i, A\Par B\} 
                               \{\Delta_j, A\Par B\} \dots}{
		\scriptInfer[d]{\dots \{\Gamma_i, A, B\}\{\Delta_j, A, B\} \dots}
                {	\scriptInfer[d]{\dots \{\Gamma_i, A, B\}\{\Delta_j, A\} \dots}
                    		{\dots \{\Gamma_i, A, B\}\{\Delta_j\} \dots}
		}
               }
          $$        
     \item[rule $(\ctimes)$: ] From ($A$ ranges over $i_2$ and $i_3$,
$X$ and $Y$ range over $i_1$ and $i_2$):
          $$\scriptInfer[\ctimes]{\dots \{\Gamma_{i_1}, X\Par Y\} 
                                  \{X\Par Y, \Gamma_{i_2}, A\ctimes B\} 
                                  \{A\ctimes B, \Gamma_{i_3}\}
                                  \{\Delta_j, A\ctimes B\} \dots}
                  {\scriptInfer[\Par]{\dots \{\Gamma_{i_1}, X\Par Y\} 
                                  \{X\Par Y, \Gamma_{i_2}, A\} 
                                  \{A, \Gamma_{i_3}\} \dots}
                         {\dots \{\Gamma_{i_1}, X, Y\} 
                                  \{X, Y, \Gamma_{i_2}, A\} 
                                  \{A, \Gamma_{i_3}\} \dots}
                   &
                   \scriptstyle \dots \{\Delta_j, B\} \dots}
          $$
          One can build:
          $$\scriptInfer[\Par]{\dots \{\Gamma_{i_1}, X\Par Y\} 
                                  \{X\Par Y, \Gamma_{i_2}, A\ctimes B\} 
                                  \{A\ctimes B, \Gamma_{i_3}\}
                                  \{\Delta_j, A\ctimes B\} \dots}
                  {\scriptInfer[\ctimes]{\dots \{\Gamma_{i_1}, X, Y\} 
                                  \{X, Y, \Gamma_{i_2}, A\ctimes B\} 
                                  \{A\ctimes B, \Gamma_{i_3}\}
                                  \{\Delta_j, A\ctimes B\} \dots}
                         {{\dots \{\Gamma_{i_1}, X, Y\} 
                                  \{X, Y, \Gamma_{i_2}, A\} 
                                  \{A, \Gamma_{i_3}\} \dots}
                           &
                           \scriptstyle \dots \{\Delta_j, B\}\dots}
                  }
          $$
     \item[rule $(\cpar)$: ] From ($X$ and $Y$ range over $i$ and $j$, $A$ and $B$ range 
     over $j$ and $k$):
          $$\scriptInfer[\cpar]{\dots \{\Delta_i, X\Par Y\} 
                                \{\Gamma_j, X\Par Y, A\cpar B\}
                                \{\Phi_k, A\cpar B\} \dots}
                  {\scriptInfer[\Par]{\dots \{\Delta_i, X\Par Y\} 
                                      \{\Gamma_j, X\Par Y, A\}
                                      \{\Gamma_j, X\Par Y, B\} 
                                      \{\Phi_k, A\} 
                                      \{\Phi_k, B\} \dots}
                         {\dots \{\Delta_i, X, Y\} 
                                \{\Gamma_j, X, Y, A\}
                                \{\Gamma_j, X, Y, B\} 
                                \{\Phi_k, A\} 
                                \{\Phi_k, B\} \dots}
                  }
          $$
          One can build:
          $$\scriptInfer[\Par]{\dots \{\Delta_i, X\Par Y\} 
                                \{\Gamma_j, X\Par Y, A\cpar B\}
                                \{\Phi_k, A\cpar B\} \dots}
                  {\scriptInfer[\cpar]{\dots \{\Delta_i, X, Y\} 
                                \{\Gamma_j, X, Y, A\cpar B\}
                                \{\Phi_k, A\cpar B\} \dots}
                         {\dots \{\Delta_i, X, Y\} 
                                \{\Gamma_j, X, Y, A\}
                                \{\Gamma_j, X, Y, B\} 
                                \{\Phi_k, A\} 
                                \{\Phi_k, B\} \dots}
                  }
          $$      
     \item[Other rules: ] Other cases are immediate.
     \end{description}
\item[\fbox{case $\cpar$}] \ \\
     \begin{description}
     \item[rule $(w)$: ] From:
          $$\scriptInfer[w]{\dots \{\Gamma_i, A\cpar B\}
                            \{\Gamma_{i_0}, A\cpar B\}
                            \{\Gamma_{i_0}, A\cpar B\} \dots}
                  {\scriptInfer[\cpar]{\dots \{\Gamma_i, A\cpar B\}
                                       \{\Gamma_{i_0}, A\cpar B\} \dots}
                         {\dots \{\Gamma_i, A\}
                                \{\Gamma_i, B\}
                                \{\Gamma_{i_0}, A\}
                                \{\Gamma_{i_0}, B\} \dots}
                  }
          $$
          One can build (twice rule w):
          $$\scriptInfer[\cpar]{\dots \{\Gamma_i, A\cpar B\}
                            \{\Gamma_{i_0}, A\cpar B\}
                            \{\Gamma_{i_0}, A\cpar B\} \dots}
                  {\scriptInfer[w]{\dots \{\Gamma_i, A\}
                                   \{\Gamma_i, B\}
                                   \{\Gamma_{i_0}, A\}
                                   \{\Gamma_{i_0}, B\}
                                   \{\Gamma_{i_0}, A\}
                                   \{\Gamma_{i_0}, B\} \dots}
                         {\scriptInfer[w]{\dots \{\Gamma_i, A\}
                                   \{\Gamma_i, B\}
                                   \{\Gamma_{i_0}, A\}
                                   \{\Gamma_{i_0}, A\}
                                   \{\Gamma_{i_0}, B\} \dots}
                         {\dots \{\Gamma_i, A\}
                                \{\Gamma_i, B\}
                                \{\Gamma_{i_0}, A\}
                                \{\Gamma_{i_0}, B\} \dots}
			 }
                 }
          $$
     \item[rule $(d)$: ] From:
          $$\scriptInfer[d]{\dots \{\Delta, X\}
                            \{\Gamma_{i_0}, A\cpar B, X\}
                            \{\Gamma_i, A\cpar B\} \dots}
                  {\scriptInfer[\cpar]{\dots \{\Delta, X\}
                                      \{\Gamma_{i_0}, A\cpar B\}
                                      \{\Gamma_i, A\cpar B\} \dots}
                        {\dots \{\Delta, X\}
                               \{\Gamma_{i_0}, A\}
                               \{\Gamma_{i_0}, B\}
                               \{\Gamma_i, A\}
                               \{\Gamma_i, B\} \dots}
                  }
          $$
          One can build (twice rule d):
          $$\scriptInfer[\cpar]{\dots \{\Delta, X\}
                                \{\Gamma_{i_0}, A\cpar B, X\}
                                \{\Gamma_i, A\cpar B\} \dots}
                  {\scriptInfer[d]{\dots \{\Delta, X\}
                                  \{\Gamma_{i_0}, A, X\}
                                  \{\Gamma_{i_0}, B, X\}
                                  \{\Gamma_i, A\}
                                  \{\Gamma_i, B\} \dots}
                        {\scriptInfer[d]{\dots \{\Delta, X\}
                                  \{\Gamma_{i_0}, A, X\}
                                  \{\Gamma_{i_0}, B\}
                                  \{\Gamma_i, A\}
                                  \{\Gamma_i, B\} \dots}
                        {\dots \{\Delta, X\}
                               \{\Gamma_{i_0}, A\}
                               \{\Gamma_{i_0}, B\}
                               \{\Gamma_i, A\}
                               \{\Gamma_i, B\} \dots}
			}
                  }
          $$
     \item[rule $(\Par)$: ] From:
          $$\scriptInfer[\Par]{\dots \{\Delta_i, X\Par Y\} 
                                \{\Gamma_j, X\Par Y, A\cpar B\}
                                \{\Phi_k, A\cpar B\} \dots}
                  {\scriptInfer[\cpar]{\dots \{\Delta_i, X, Y\} 
                                \{\Gamma_j, X, Y, A\cpar B\}
                                \{\Phi_k, A\cpar B\} \dots}
                         {\dots \{\Delta_i, X, Y\} 
                                \{\Gamma_j, X, Y, A\}
                                \{\Gamma_j, X, Y, B\} 
                                \{\Phi_k, A\} 
                                \{\Phi_k, B\} \dots}
                  }
          $$    
          One can build:
          $$\scriptInfer[\cpar]{\dots \{\Delta_i, X\Par Y\} 
                                \{\Gamma_j, X\Par Y, A\cpar B\}
                                \{\Phi_k, A\cpar B\} \dots}
                  {\scriptInfer[\Par]{\dots \{\Delta_i, X\Par Y\} 
                                      \{\Gamma_j, X\Par Y, A\}
                                      \{\Gamma_j, X\Par Y, B\} 
                                      \{\Phi_k, A\} 
                                      \{\Phi_k, B\} \dots}
                         {\dots \{\Delta_i, X, Y\} 
                                \{\Gamma_j, X, Y, A\}
                                \{\Gamma_j, X, Y, B\} 
                                \{\Phi_k, A\} 
                                \{\Phi_k, B\} \dots}
                  }
          $$
      \item[rule $(\otimes)$: ] From:
          $$\scriptInfer[\otimes]{\dots \{\Gamma_i, A\cpar B\}
                                 \{\Delta_j, A\cpar B, \Phi_k, X\otimes 
                                 Y\} \dots}
                  {\scriptInfer[\cpar]{\dots \{\Gamma_i, A\cpar B\}
                                       \{\Delta_j, A\cpar B, X\}
                                       \dots} 
                         {\dots \{\Gamma_i, A\}
                                \{\Gamma_i, B\}
                                \{\Delta_j, A, X\}
                                \{\Delta_j, B, X\}
                          \dots}
                  &
                  \scriptstyle \dots \{\Phi_k, Y\} \dots
                  }
         $$
         One can build:
         $$\scriptInfer[\cpar]{\dots \{\Gamma_i, A\cpar B\}
                                 \{\Delta_j, A\cpar B, \Phi_k, X\otimes 
                                 Y\} \dots}
                  {\scriptInfer[\otimes]{\dots \{\Gamma_i, A\}
                                        \{\Gamma_i, B\}
                                        \{\Delta_j, A, \Phi_k, X\otimes Y\}
                                        \{\Delta_j, B, \Phi_k, X\otimes Y\} \dots}
                         {{\dots \{\Gamma_i, A\}
                                \{\Gamma_i, B\}
                                \{\Delta_j, A, X\}
                                \{\Delta_j, B, X\} \dots}
                           &
                           \scriptstyle \dots \{\Phi_k, Y\} \dots}
                  }
         $$
      \item[rule $(\cpar)$: ] From:
         $$\scriptInfer[\cpar]{\dots \{\Gamma_i, X\cpar Y\}
                               \{\Delta_j, X\cpar Y, A \cpar B\}
                               \{\Phi_k, A\cpar B\} \dots}
                 {\scriptInfer[\cpar]{\dots \{\Gamma_i, X\}
                                      \{\Gamma_i, Y\}
                                      \{\Delta_j, X, A \cpar B\}
                                      \{\Delta_j, Y, A \cpar B\}
                                      \{\Phi_k, A\cpar B\} \dots}
                        {\dots \{\Gamma_i, X\}
                               \{\Gamma_i, Y\}
                               \{\Delta_j, X, A\}
                               \{\Delta_j, X, B\}
                               \{\Delta_j, Y, A\}
                               \{\Delta_j, Y, B\}
                               \{\Phi_k, A\}
                               \{\Phi_k, B\} \dots}
                 }
         $$
         One can build:
         $$\scriptInfer[\cpar]{\dots \{\Gamma_i, X\cpar Y\}
                               \{\Delta_j, X\cpar Y, A \cpar B\}
                               \{\Phi_k, A\cpar B\} \dots}
                 {\scriptInfer[\cpar]{\dots \{\Gamma_i, X\cpar Y\}
                                      \{\Delta_j, X \cpar Y, A\}
                                      \{\Delta_j, X \cpar Y, B\}
                                      \{\Phi_k, A\}
                                      \{\Phi_k, B\} \dots}
                        {\dots \{\Gamma_i, X\}
                               \{\Gamma_i, Y\}
                               \{\Delta_j, X, A\}
                               \{\Delta_j, X, B\}
                               \{\Delta_j, Y, A\}
                               \{\Delta_j, Y, B\}
                               \{\Phi_k, A\}
                               \{\Phi_k, B\} \dots}
                 }
         $$
      \item[rule $(\ctimes)$: ] From:
         $$\scriptInfer[\ctimes]{\dots \{\Gamma_{i_1}, X\ctimes Y\}
                                 \{\Gamma_{i_2}, X\ctimes Y, A \cpar B\}
                                 \{\Delta_{j}, X\ctimes Y\}
                                 \{\Phi_k, A\cpar B\} \dots}
                 {\scriptInfer[\cpar]{\dots \{\Gamma_{i_1}, X\}  
                                      \{\Gamma_{i_2}, X, A \cpar B\}
                                      \{\Phi_k, A\cpar B\} \dots}
                        {\dots \{\Gamma_{i_1}, X\}
                               \{\Gamma_{i_2}, X, A\}
                               \{\Gamma_{i_2}, X, B\}
                               \{\Phi_k, A\}
                               \{\Phi_k, B\} \dots}
                 &
                 \scriptstyle \dots \{\Delta_{j}, Y\} \dots
                 }
         $$
         One can build:
         $$\scriptInfer[\cpar]{\scriptstyle \dots \{\Gamma_{i_1}, X\ctimes Y\}
                                 \{\Delta_{j}, X\ctimes Y\}
                                 \{\Gamma_{i_2}, X\ctimes Y, A \cpar B\}
                                 \{\Phi_k, A\cpar B\} \dots}
                 {\scriptInfer[\ctimes]{\dots \{\Gamma_{i_1}, X \ctimes Y\}
                               \{\Delta_{j}, X\ctimes Y\}
                               \{\Gamma_{i_2}, X\ctimes Y, A\}
                               \{\Gamma_{i_2}, X\ctimes Y, B\}
                               \{\Phi_k, A\}
                               \{\Phi_k, B\} \dots}
                        {\dots \{\Gamma_{i_1}, X\}
                               \{\Gamma_{i_2}, X, A\}
                               \{\Gamma_{i_2}, X, B\}
                               \{\Phi_k, A\}
                               \{\Phi_k, B\} \dots
                        &
                        \scriptstyle \dots \{\Delta_{j}, Y\} \dots
                        }
                 }
         $$
      \end{description}
\item[\fbox{case other}] Other cases are treated as usual.
\end{description}

\end{proof}


\begin{lemma}[Synchrony of the cut rule]
The cut rule is synchronous, i.e. let a proof of $\mathcal S$ be of the following form (R is a rule):
$$\scriptInfer[cut]{\mathcal S}
	{
	\scriptInfer[R]{{\mathcal W}_1[A]}
		{{\mathcal U}[A]}
	&
	\scriptstyle {\mathcal V}[A^\perp]
	}
$$
one can build a proof of $\mathcal S$ of the form:
$$\scriptInfer[R]{\mathcal S}
	{
	\scriptInfer[cut]{{\mathcal W}_2}
		{
		{\mathcal U}[A]
		&
		\scriptstyle {\mathcal V}[A^\perp]
		}
	}
$$
Moreover, the height of the partial proof ending with the cut rule in the second case is less than in the first case.
\end{lemma}

\begin{proof}
By proving permutation properties as for proving synchrony of $\otimes$.
\end{proof}


\begin{proposition}
The system enjoys cut-elimination: if $\mathcal S$ is a provable multi-sequent, 
then there exists at least one cut-free proof of $\mathcal S$.
\end{proposition}

\begin{proof} By induction on the height of the proof. The synchrony of cut 
allows us to check only the last rule applied on each branch. Furthermore,
as the connectives $\Par$, $\cpar$ and $\with$ are asynchronous, we are
allowed to consider that only dual rules are applied as last rules.
\begin{description}
\item[\fbox{case Axiom}] Obvious.
\item[\fbox{case $1$/$\bot$} ] This case is obvious as the proof looks like:
$$\scriptInfer[cut]{\dots \{\Gamma_i\} \dots}
        { \scriptInfer[\bot]{\dots \{\bot, \Gamma_i\} \dots}
                {\dots \{\Gamma_i\} \dots}
        &
          \scriptInfer[1]{\{1\}}{}
        }
$$
\item[\fbox{case $\otimes$/$\Par$} ] From:
$$\scriptInfer[cut]{\dots \{\Phi_k, \Gamma_i, \Delta_j\} \dots}
        { \scriptInfer[\Par]{\dots \{A\Par B, \Phi_k\} \dots}
                {\dots \{A, B, \Phi_k\} \dots}
        &
          \scriptInfer[\otimes]{\dots \{A^\perp \otimes B^\perp, \Gamma_i, \Delta_j\} \dots}
                {\dots \{A^\perp, \Gamma_i\} \dots
                 &
                 \scriptstyle \dots \{B^\perp, \Delta_j\} \dots}
        }
$$
One can build:
$$\scriptInfer[cut]{\dots \{\Phi_k, \Gamma_i, \Delta_j\} \dots}
        { \scriptInfer[cut]{\dots \{B, \Phi_k, \Gamma_i\} \dots}
                {\dots \{A, B, \Phi_k\} \dots
                &
                \scriptstyle \dots \{A^\perp, \Gamma_i\} \dots}
        &
          \scriptstyle \dots \{B^\perp, \Delta_j\} \dots
        }
$$
\item[\fbox{case $\ctimes$/$\cpar$} ] From:
$$\scriptInfer[cut]{\dots \{\Phi_k, \Gamma_i\} \{\Phi_k, \Delta_j\} \dots}
        { \scriptInfer[\cpar]{\dots \{A\cpar B, \Phi_k\} \dots}
                {\dots \{A, \Phi_k\} \{B, \Phi_k\} \dots}
        &
          \scriptInfer[\ctimes]{\dots \{A^\perp \ctimes B^\perp, \Gamma_i\}
                          \{A^\perp \ctimes B^\perp, \Delta_j\} 
                          \dots}
                {\dots \{A^\perp, \Gamma_i\} \dots
                 &
                 \scriptstyle \dots \{B^\perp, \Delta_j\} \dots}
        }
$$
One can build:
$$\scriptInfer[cut]{\dots \{\Phi_k, \Gamma_i\} \{\Phi_k, \Delta_j\} \dots}
        {\scriptInfer[cut]{\dots \{\Gamma_i, \Phi_k\} \{B, \Phi_k\} \dots}
               {\dots \{A^\perp, \Gamma_i\} \dots
               &
                \scriptstyle \dots \{A, \Phi_k\} \{B, \Phi_k\} \dots}
        &
         \scriptstyle \dots \{B^\perp, \Delta_j\} \dots
        }
$$
\item[\fbox{case $.^s/.^u$}] From:
     $$\scriptInfer[{\rm cut}]{\dots \{\Gamma_i, \Delta_j^s\}
				\{\Gamma_i, \Delta_{j'}^s, \Phi_k^s\} \dots}
		{ \scriptInfer[.^s]{\dots \{\Gamma_i, \Delta_j^s\}
				\{\Gamma_i, \Delta_{j'}^s, A^s\} \dots}
			{\dots \{\Gamma_i^j, \Delta_j\}
				\{\Gamma_i^{j'}, \Delta_{j'}, A\} \dots}
		&
		  \scriptInfer[.^u]{\dots \{A^{\perp u}, \Phi_k^s\} \dots}
			{\dots \{A, \Phi_k\} \dots}
		}
     $$
     One can infer:
     $$\scriptInfer[.^s]{\dots \{\Gamma_i, \Delta_j^s\}
				\{\Gamma_i, \Delta_{j'}^s, \Phi_k^s\} \dots}
		{ \scriptInfer[{\rm cut}]{\dots \{\Gamma_i^j, \Delta_j\}
				\{\Gamma_i^{j'}, \Delta_{j'}, \Phi_k\} \dots}
			{\dots \{\Gamma_i^j, \Delta_j\}
				\{\Gamma_i^{j'}, \Delta_{j'}, A\} \dots
			&
		  	  \scriptstyle \dots \{A, \Phi_k\} \dots
		}	}
     $$
\end{description}
\end{proof}